\documentclass[nofootinbib, reprint, amsmath,amssymb, aps]{revtex4-2}
\usepackage{amsthm}
\newtheorem{theorem}{Theorem}[section]
\newtheorem{proposition}[theorem]{Proposition}
\newtheorem{corollary}[theorem]{Corollary}

\usepackage{graphicx}
\usepackage{bm}
\usepackage[colorlinks]{hyperref}
\usepackage{stackrel}
\usepackage{subcaption}
\usepackage{listings}
\usepackage{courier}

\newcommand{\be}{\begin{equation}}
\newcommand{\ee}{\end{equation}}
\newcommand{\bea}{\begin{eqnarray}}
\newcommand{\eea}{\end{eqnarray}}

\begin{document}

\title{Statistical Flux Freezing with Magnetic Path-lines in Turbulence}

\author{Amir Jafari}
\email{elenceq@jhu.edu}
\begin{abstract}
The classical Alfv\'en flux--freezing theorem states that, in ideal
magnetohydrodynamics, magnetic flux is transported by the plasma flow.
In turbulent plasmas the velocity and magnetic fields can be spatially
rough, so the regularity assumptions supporting deterministic transport
need not hold. Earlier stochastic flux-freezing formulations use random
fluid trajectories to represent magnetic transport. Motivated by the
magnetic path-line description, we study the stochastic regularization
of $\dot{\mathbf{X}}(t)=\mathbf{B}(\mathbf{X}(t),t)$ for a prescribed
divergence-free field. Under a persistent backward pair-dispersion
assumption, every weak subsequential zero-noise limit of the endpoint
law is non-Dirac. With path-space tightness and continuity of the limiting
drift, the limiting laws are supported on nonunique path-line histories.
This kinematic statement is distinguished from magnetic-flux transport,
which additionally requires compatibility with the induction equation.
For incompressible, constant-density resistive MHD, with $\mathbf{B}$
in Alfv\'en-speed units, the Els\"asser fields
$\mathbf{z}^{\pm}=\mathbf{u}\pm\mathbf{B}$ supply compatible drifts
because $\mathbf{z}^{\pm}\times\mathbf{B}=\mathbf{u}\times\mathbf{B}$.
The resulting stochastic flux identity is exact for smooth fields at
positive magnetic diffusivity. The flux identity passes to the ideal limit
under the stated convergence and uniform-integrability assumptions.
Persistent pair dispersion for the chosen Els\"asser trajectories
separately ensures non-Dirac trajectory limits. The analysis thus gives a conditional noncollapse criterion and
an induction-compatible path-line representation of statistical flux
freezing, without deriving the dispersion assumption from roughness alone.
\end{abstract}

\maketitle

\section{Introduction}

The classical Alfv\'en flux--freezing theorem~\cite{Alfven1942} in
\emph{ideal} magnetohydrodynamics (MHD), describing the transport of
magnetic flux by the flow, has played a central role in plasma physics.
At large kinetic and magnetic Reynolds numbers, however, the flow can
be turbulent. The
classical derivation of flux freezing relies on implicit regularity
assumptions for the velocity and magnetic fields. In particular, it
assumes that Lagrangian trajectories of the flow are uniquely defined; e.g., see~\cite{Eyink2011} and references therein.
In turbulence this assumption can easily fail, as the velocity field
may become only H\"older continuous with exponent $h = 1/3 < 1$, well
below the Lipschitz regularity sufficient in the Picard--Lindel\"of
theorem for uniqueness of trajectories~\cite{Sreenivasan1984,Eyink2018}.
This lack of regularity is closely related to the Onsager
conjecture~\cite{Onsager1945,EyinkSreenivasan2006,Eyink2018} concerning
anomalous dissipation in turbulent flows.

Roughness can permit nonunique solutions of the trajectory equation.
\emph{Spontaneous stochasticity} means that a non-Dirac trajectory law
persists as the regularizing noise vanishes; nonuniqueness alone does
not guarantee this~\cite{Bernard1998}. Instead of converging to a single
deterministic trajectory, the vanishing--noise limit can produce a
probability distribution of possible particle histories. In the context
of magnetized turbulence, this behavior has important implications for
magnetic transport and flux freezing~\cite{Eyink2011}. The stochastic representation of Ref.~\cite{Eyink2011} already uses
time-parametrized fluid trajectories to transport magnetic vectors and
surfaces. Magnetic field-line wandering provides a related geometric
picture of turbulent reconnection~\cite{Review2020}.

Field lines, however, are geometric objects defined only at a fixed
instant of time and require an auxiliary curve parameter. As a result
they do not naturally form dynamical objects evolving in time.
Especially in turbulent plasmas, field lines do not preserve any
identity as the magnetic field evolves~\cite{Review2020}.

In this paper we study \emph{magnetic path-lines}, motivated by
Ref.~\cite{Jafari2025}. These are trajectories $\mathbf{X}(t)$
generated by a prescribed time-dependent field $\mathbf{B}(\mathbf{x},t)$:
\begin{equation}
\dot{\mathbf{X}}(t)=\mathbf{B}(\mathbf{X}(t),t).
\label{path-lineeq}
\end{equation}

Here $\mathbf{B}$ has velocity units; when it denotes the physical
magnetic field, it is expressed in Alfv\'en-speed units at constant
density. In the kinematic analysis, $\mathbf{B}$ is simply prescribed
and divergence-free. It may be a magnetic field or a regularized
transport field, but no source-free flux law for that drift is assumed.
Unlike instantaneous field lines, the curves in
Eq.~\eqref{path-lineeq} sample the field at successive physical times.
Their stochastic properties can therefore be studied using Lagrangian
dynamics and stochastic processes~\cite{Jafari2025,Jafari2025MHD}.

The present work separates two questions: noncollapse of regularized
path-lines, and transport of the physical magnetic flux. Persistent
backward pair dispersion excludes deterministic subsequential limits of
the chosen trajectory law. A magnetic-flux identity additionally needs
an induction-compatible drift; this is supplied below by
$\mathbf{z}^{\pm}=\mathbf{u}\pm\mathbf{B}$, rather than by an
arbitrary magnetic-path-line field.

The path-line viewpoint of Ref.~\cite{Jafari2025} provides a framework
for interpreting magnetic reconnection. In that formulation the separation of nearby
path-lines naturally measures the rate at which magnetic structures
diverge in spacetime. In particular, the rate of separation of
path-lines is consistent with reconnection rates predicted by models
based on field-line diffusion. In laminar flows this reproduces the
Sweet--Parker reconnection rate

\[
v_{\rm rec}\sim \frac{v_A}{\sqrt{S}},
\]

where $v_A$ is the Alfv\'en speed and $S$ is the Lundquist number
\cite{Parker1957,Sweet1958}. For sub-Alfv\'enic turbulence with comparable injection and reconnection-layer
lengths, the corresponding scaling traditionally obtained using
field-line wandering is

\[
v_{\rm rec}\sim v_A\,M_A^2 ,
\]

where $M_A$ is the Alfv\'enic Mach number
\cite{LazarianandVishniac1999,Eyinketal2013,Lazarianetal2015,Lalescuetal2015,Review2020}.
This formalism therefore captures the same underlying stochastic
transport physics that appears in field-line formulations, while
providing a dynamical trajectory framework in spacetime~\cite{Jafari2025}.

Appendix~\ref{App0} distinguishes phase-space compression from magnetic
topology change and discusses the heuristic rate $\delta u(l)/l$.
Neither supplies a proof of the pair-dispersion hypothesis used below.

We will not pursue these dynamical considerations about reconnection and
topology change in the present work. Instead we focus on flux freezing
itself and analyze it directly in terms of path-lines. To investigate
the behavior of these trajectories in rough (turbulent) transport
fields, we introduce a stochastic regularization of
Eq.~(\ref{path-lineeq}) and examine the limit of vanishing regularizing
diffusivity. Its identification with magnetic diffusivity is made only
in the induction-compatible application of Sec.~\ref{sec:flux}. The central question is whether the regularized stochastic
trajectories collapse to a single deterministic path-line as the noise
amplitude tends to zero, or whether a finite spread of possible
histories persists.

A useful diagnostic of this behavior is backward pair dispersion.
Consider two trajectories that arrive at the same spacetime point
$(\mathbf{x},t)$ and are integrated backward in time. If the dynamics
were deterministic in the ideal limit, the two trajectories would
coincide at all earlier times. Persistent separation of such
trajectories therefore signals nonuniqueness of path histories and
provides a natural indicator of spontaneous stochasticity.

Our analysis shows that persistent backward pair dispersion excludes
deterministic subsequential limits of the chosen path-line law.
Path-space tightness and continuity of the limiting drift then identify
the limiting histories as solutions of the zero-noise trajectory equation.
These statements are conditional: roughness alone does not prove
persistent dispersion.

For physical magnetic-flux transport we use the Els\"asser drift
$\mathbf{z}^{\pm}$, giving an exact finite-diffusivity stochastic flux
identity. The flux identity passes to the ideal limit under the stated
convergence and uniform-integrability assumptions of Sec.~\ref{sec:flux}.
Persistent pair dispersion for the chosen Els\"asser trajectories
separately ensures non-Dirac trajectory limits. This is a specialization of the stochastic
transport representation~\cite{Eyink2011}, not a replacement of an
induction law by a purely magnetic advection equation.

Section~\ref{II} introduces the stochastic path-line formulation and
presents a physically motivated derivation of the main result.
Technical assumptions and detailed proofs are collected in
Appendix~\ref{App1}. A related trajectory formulation arises in the
Els{\"a}sser representation of incompressible MHD, where the dynamics is
expressed in terms of the fields \(\mathbf{z}^\pm = \mathbf{u} \pm
\mathbf{B}\). One may then consider trajectories generated by the
Els{\"a}sser fields,

\[
\dot{\mathbf{X}}^\pm = \mathbf{z}^\pm(\mathbf{X}^\pm,t),
\]

which are the transport directions relevant to the flux application.
Appendix~\ref{App2} relates them to the cross-advection in the Els\"asser
equations. The physical implications are discussed in
Sec.~\ref{Discussion}, and Appendix~\ref{sec:feynman-kac-review}
explains the scalar and vector stochastic representations.

\section{Stochastic Path-lines and Flux Freezing}\label{II}

We now develop the stochastic path-line formulation underlying the
preceding discussion. Our aim is to replace deterministic magnetic
path-lines by a regularized stochastic dynamics and then analyze the
zero--noise limit in a form directly analogous to stochastic
formulations of turbulent transport.

We consider trajectories defined by
\begin{equation}
\dot{\mathbf{X}}(t) = \mathbf{B}(\mathbf{X}(t),t),
\label{detpath}
\end{equation}
where $\mathbf{B}(\mathbf{x},t)$ is a divergence--free vector field on
$\mathbb{R}^3$; $\nabla\cdot \mathbf{B} = 0$. 

Sections II A--C treat $\mathbf{B}$ as a prescribed divergence-free
path-line drift; no induction equation is assumed for that kinematic
analysis. In Sec.~\ref{sec:flux} and Appendix~\ref{App2}, $\mathbf{B}$
denotes the physical magnetic field in Alfv\'en-speed units, and the
Els\"asser drift is $\mathbf{z}^{\pm}=\mathbf{u}\pm\mathbf{B}$.

If $\mathbf{B}$ is Lipschitz in space, the Picard--Lindel{\"o}f theorem
guarantees existence and uniqueness of solutions of \eqref{detpath}.
However, turbulent vector fields may violate this regularity
condition. For example, Kolmogorov scaling~\cite{Kolmogorov1941} suggests spatial H{\"o}lder
continuity with exponent $1/3$, which lies below the Lipschitz
threshold. The divergence--free condition implies that the
advection operator $\mathbf{B}\cdot\nabla$ is formally skew-adjoint in $L^2$
and that sufficiently regular deterministic flows preserve volume.

\subsection{Stochastic regularization and transport}

To regularize the dynamics we introduce a small stochastic perturbation.
For a constant $\kappa>0$ consider the backward stochastic differential equation
obtained by adding a diffusive perturbation to the path-line dynamics,
\begin{equation}
\begin{aligned}
\widehat d\mathbf{X}_s^\kappa
&=\mathbf{B}(\mathbf{X}_s^\kappa,s)\,ds
+\sqrt{2\kappa}\,\widehat d\mathbf{W}_s,\\
\mathbf{X}_t^\kappa&=\mathbf{x}.
\end{aligned}
\label{SDE}
\end{equation}
Here $\mathbf{W}_s$ is standard three-dimensional Brownian motion,
$\widehat d$ denotes the backward It\^o differential, and physical time
$s$ decreases from $t$. The condition $\mathbf{X}_t^\kappa=\mathbf{x}$
prescribes the terminal value of this backward process; it is not a
conditioning operation on a forward process with an unspecified prior.
The noise is additive, so the backward It\^o and Stratonovich forms of
the trajectory equation coincide. The regularizing diffusivity $\kappa$
is identified with the magnetic diffusivity $\eta$ only in
Sec.~\ref{sec:flux}.

The regularized dynamics induces a natural backward transport
representation. For $f\in C_b^2(\mathbb{R}^3)$, define
\begin{equation}
u^\kappa(\mathbf{x},t)
=
\mathbb E\!\left[f(\mathbf{X}_\tau^\kappa)\mid \mathbf{X}_t^\kappa=\mathbf{x}\right],
\qquad \tau<t.
\label{rep_main}
\end{equation}
The conditional expectation is understood through the transition
kernel $p^\kappa(\tau,\mathbf{y}|t,\mathbf{x})$ of the diffusion
process. Equivalently, in terms of the transition kernel
$p^\kappa(\tau,\mathbf{y}|t,\mathbf{x})$,
\begin{equation}
u^\kappa(\mathbf{x},t)
=
\int f(\mathbf{y})\,p^\kappa(\tau,\mathbf{y}|t,\mathbf{x})\,d\mathbf{y}.
\label{rep_kernel_main}
\end{equation}
Standard backward Kolmogorov theory then gives
\begin{equation}
\partial_t u^\kappa+\mathbf{B}\cdot\nabla u^\kappa
=\kappa\Delta u^\kappa,
\qquad u^\kappa(\mathbf{x},\tau)=f(\mathbf{x}).
\label{kol_main}
\end{equation}
Thus stochastic magnetic path-lines transport observables according to
an advection--diffusion equation with drift $\mathbf{B}$ and diffusivity $\kappa$, posed forward in $t$
with initial data at $t=\tau$. The trajectories themselves run backward
in physical time.
Physically, this equation describes the stochastic spreading of
backward path-line histories generated by the effective transport
field.

\subsection{Pair dispersion and noncollapse}

The crucial question is whether the stochastic regularization
collapses to a single deterministic path-line as $\kappa\to0$.
To test this, consider two trajectories
$\mathbf{X}_s^{\kappa,1}$ and $\mathbf{X}_s^{\kappa,2}$ arriving at
the same final spacetime point $(\mathbf{x},t)$. If the zero--noise
limit were deterministic, then these two backward trajectories would
converge to the same point at any earlier time $\tau<t$.

Persistent pair dispersion rules out this possibility.
If two trajectories conditioned to arrive at the same final point
retain a finite probability of separation at earlier times, the
zero--noise limit cannot collapse to a single deterministic path-line.
Here $\mathbf{X}^{\kappa,1}$ and $\mathbf{X}^{\kappa,2}$ denote two
independent realizations of the conditioned stochastic dynamics,
equivalently two independent samples from the transition kernel
$p^\kappa(\tau,\cdot\mid t,\mathbf{x})$.

\begin{theorem}[Noncollapse of the zero--noise limit]
\label{thm:main}
Let $p^\kappa(\tau,\cdot|t,\mathbf{x})$ denote the backward transition
kernel associated with the stochastic dynamics \eqref{SDE}, and let
$\mathbf{X}^{\kappa,1}$ and $\mathbf{X}^{\kappa,2}$ be two independent
samples drawn from this conditional distribution. Assume that the family
$\{p^\kappa(\tau,\cdot|t,\mathbf{x})\}_{0<\kappa\le\kappa_0}$ is tight on
$\mathbb{R}^3$ for some fixed $\kappa_0>0$. Suppose there exist
$\tau<t$, $r_*>0$, and $p_*>0$ such
that
\begin{equation}
\liminf_{\kappa\to0}
\mathbb P\big(
|\mathbf{X}_\tau^{\kappa,1}-\mathbf{X}_\tau^{\kappa,2}|\ge r_*
\big)
\ge p_* .
\label{disp}
\end{equation}
Such persistent pair dispersion is expected in turbulent flows with
rough advecting fields, where backward trajectories separate according
to Richardson--type scaling (see e.g.\
\cite{Bernard1998,Eyink2011}). Then the transition kernel
$p^\kappa(\tau,\cdot|t,\mathbf{x})$ cannot converge weakly to a Dirac
measure as $\kappa\to0$:
\begin{equation}
p^\kappa(\tau,\cdot|t,\mathbf{x})\nRightarrow \delta_{\mathbf{a}_*}.
\label{noncollapse_statement}
\end{equation}
Every weak subsequential limit is non-Dirac; uniqueness of a limiting
law is not asserted. The conclusion concerns the chosen path-line law
at the specified terminal point and earlier time.
\end{theorem}

\begin{proof}[Proof sketch.]
Suppose some sequence $\kappa_n\downarrow0$ has
$p^{\kappa_n}(\tau,\cdot|t,\mathbf{x})\Rightarrow\delta_{\mathbf{a}_*}$.
Both endpoint marginals then concentrate at $\mathbf{a}_*$, so their
separation tends to zero in probability. This contradicts
\eqref{disp}, whose lower bound applies along every sequence.
Appendix~\ref{App1} gives the detailed argument.
\end{proof}

With $\limsup$ instead of $\liminf$ in \eqref{disp}, the same argument
excludes convergence of the full family to a Dirac measure, but need not
exclude deterministic subsequences. We use the stronger assumption to
obtain the stated subsequential conclusion.

Note that this theorem is conditional: it does not establish the
existence of persistent pair dispersion for the path-line dynamics
itself, but shows that whenever such dispersion occurs the stochastic
regularization cannot collapse to a deterministic trajectory. In
turbulent flows with rough advecting fields such persistent backward
dispersion is widely expected on physical grounds and is consistent
with Richardson-type separation observed in turbulent transport.

\subsection{Stochastic path-line freezing}

For the following path-space statement, also assume the tightness and
continuity conditions in Appendix~\ref{App1}. Theorem~\ref{thm:main}
then excludes Dirac path-space limits whenever \eqref{disp} holds.
We use the term \emph{stochastic path-line freezing} for this kinematic
persistence of random histories, not as an independent magnetic-flux law.

\begin{corollary}[Stochastic path-line freezing]
Under these assumptions, every weak subsequential path-space limit is a
non-Dirac probability measure supported on solutions of the zero-noise
path-line equation with the prescribed terminal value. Different
subsequences may select different measures.
\end{corollary}

A physical statement about magnetic connectivity or reconnection requires
compatibility with the induction equation and additional dynamical input.

\subsection{Statistical flux freezing}\label{sec:flux}

Consider incompressible, constant-density resistive MHD, with the
physical magnetic field $\mathbf{B}$ in Alfv\'en-speed units and
$\kappa=\eta>0$, where $\eta$ denotes magnetic diffusivity. For either
fixed sign, $\mathbf{z}^{\pm}=\mathbf{u}\pm\mathbf{B}$ satisfies
\[
\begin{aligned}
\mathbf{z}^{\pm}\times\mathbf{B}&=\mathbf{u}\times\mathbf{B},\\
\partial_t\mathbf{B}
&=\nabla\times(\mathbf{z}^{\pm}\times\mathbf{B})
  +\kappa\Delta\mathbf{B}.
\end{aligned}
\]
Thus Appendix~\ref{AppD} applies with $\mathbf{H}=\mathbf{B}$ and
$\mathbf{V}=\mathbf{z}^{\pm}$. The equality uses the freedom to add a
field-parallel velocity to the induction drift; it is not a new
induction equation. The MHD fields are held fixed while averaging over
the auxiliary Brownian noise.

Let $\widetilde{\mathbf{A}}_{s,t}$ denote the corresponding backward
stochastic flow \eqref{AppD:backwardSDE}, and write
$\mathbf{B}_0=\mathbf{B}(\cdot,t_0)$. Within a flow realization all
surface points use the same Brownian motion. For a compact oriented
smooth surface $S_t$, Appendix~\ref{AppD} gives

\begin{equation}
\int_{S_t}\mathbf{B}(\mathbf{x},t)\cdot d\mathbf{A}
=
\mathbb{E}\!\left[
\int_{\widetilde{\mathbf{A}}_{t_0,t}(S_t)}
\mathbf{B}_0(\mathbf{a})\cdot d\mathbf{A}
\right],
\label{stoch-alfven-main}
\end{equation}

Equation~\eqref{stoch-alfven-main} is exact for smooth, unfiltered
resistive MHD at each $\kappa>0$. Theorem~\ref{thm:main} applies to
this representation only when \eqref{disp} holds for the chosen
$\mathbf{z}^{\pm}$-generated trajectories. Dispersion of
$\mathbf{B}$-generated trajectories alone does not establish that
condition.

For a singular ideal-limit statement, allow the fields and initial data
to depend on $\kappa$ and approach a rough limit, suppressing this
dependence in \eqref{stoch-alfven-main}. Assume that its left-hand flux
converges to $\Phi_t$, that its random initial-flux integral converges
in distribution to a random variable $\Phi_{t_0}$, and that these
random integrals are uniformly integrable. Then
$\Phi_t=\mathbb{E}\Phi_{t_0}$. This limiting flux identity does not by
itself assert the existence of a classical limiting random surface.
Endpoint noncollapse alone supplies neither flux convergence nor
control of surface deformation.

A fixed smooth drift has deterministic zero-noise collapse. With
suitable convergence of flows and surface elements, the stochastic
identity reduces to the usual deterministic transported-flux relation
in that regular limit. In a rough limit, noncollapse concerns the
chosen trajectory representation and does not exclude every possible
deterministic flux-transport representation. The stochastic
representation itself is the established one~\cite{Eyink2011}; the
present specialization uses the induction-compatible Els\"asser drift.

\section{Discussion}\label{Discussion}

The formulation separates the stochasticity of a prescribed path-line
law from the transport of physical magnetic flux. Earlier stochastic
flux-freezing work already uses fluid trajectories~\cite{Eyink2011}.
Here the kinematic analysis keeps the magnetic-path-line notation
$\dot{\mathbf{X}}=\mathbf{B}(\mathbf{X},t)$, while the physical flux
representation uses $\mathbf{z}^{\pm}=\mathbf{u}\pm\mathbf{B}$.
The distinction is the choice of drift, not the introduction of
time-parametrized trajectories into stochastic flux freezing.

The noncollapse theorem is conditional on persistent backward pair
dispersion. It is consistent with the trajectory-dispersion mechanism
studied in MHD simulations~\cite{Eyinketal2013,Review2020}, but does not
by itself prove reconnection, determine its rate, or exclude a different
deterministic representation of a particular magnetic field.

In the local Alfv\'en-wave approximation, disturbances propagate
relative to the plasma along $\pm\mathbf{B}$, so the corresponding
laboratory-frame directions are $\mathbf{z}^{\pm}$. In a general
nonlinear, nonuniform flow, the path-line equations are not complete
characteristic equations for wave amplitudes: pressure and
cross-advection remain in the Els\"asser dynamics. The exact reason
for using $\mathbf{z}^{\pm}$ in the flux theorem is the induction
identity in Sec.~\ref{sec:flux}, not a wave-packet approximation.
For filtered MHD, the subscale electromotive-force term must also be
retained; the unfiltered flux formula cannot simply be applied to
$\mathbf{u}_l\pm\mathbf{B}_l$ with that term omitted.

Persistent backward pair dispersion excludes deterministic
subsequential limits of the chosen path-line law. For the Els\"asser
drift this accompanies an exact finite-diffusivity magnetic-flux
identity. Its singular ideal-limit interpretation remains conditional
on the separate convergence and integrability assumptions in
Sec.~\ref{sec:flux}. No inference that flux is conserved \emph{only}
statistically for every possible representation follows from endpoint
noncollapse alone.

The stochastic path-line picture also has implications for magnetic
reconnection. Classical reconnection theories often assume that
magnetic connectivity evolves through localized nonideal effects,
while the large-scale magnetic field remains approximately frozen
into the flow. In turbulent plasmas, however, the stochastic
transport implied by spontaneous stochasticity allows magnetic
connectivity to evolve through turbulent dispersion of magnetic
structures. Recent work based on the path-line formulation has
derived expressions for rates of magnetic connectivity change in
turbulent flows~\cite{Jafari2025}. The present results provide a
complementary foundation for such approaches by establishing that the
underlying path-line dynamics may remain intrinsically stochastic in
the ideal limit.

Within this viewpoint, reconnection is associated with rapid
reorganization of magnetic pattern rather than with the identity of
individual field lines. A quantitative theory of reconnection rates requires coupling the
present path-line framework to the full induction equation and the
dynamics of the magnetic and velocity fields. The analysis presented
here establishes the kinematic ingredient of this picture by showing
that the ideal-limit path-line dynamics can remain intrinsically
stochastic.

The persistent pair-dispersion condition can be tested in numerical
simulations. For a fixed Eulerian field realization, integrate two
independent backward stochastic trajectories from $(\mathbf{x},t)$.
Individual endpoint samples estimate $p^\kappa(\tau,\cdot|t,\mathbf{x})$;
pair separations estimate the probability in \eqref{disp}, not the
single-particle kernel itself. A positive lower bound uniform as
$\kappa\to0$ tests the strengthened noncollapse hypothesis. For the
physical flux representation, the drift must be the chosen
$\mathbf{z}^{\pm}$. The field resolution and diffusivities must be
varied consistently with a rough-field limit: decreasing noise at a
fixed smooth resolved field instead gives deterministic collapse.

The path-line viewpoint also highlights the difference between
time-evolving trajectories and instantaneous geometric objects.
Magnetic field lines are defined only at a fixed time and require an
auxiliary curve parameter, whereas path-lines are ordinary dynamical
trajectories in spacetime. Treating these trajectories as the basic
objects provides a direct way to analyze stochastic transport
mechanisms arising from rough advecting fields.

The noncollapse proof is kinematic, and the finite-diffusivity flux
identity is conditional on prescribed fields satisfying induction.
Neither determines energy conversion or a reconnection rate. Those
questions require the dynamics of both the magnetic and velocity
fields. The Els\"asser extension in Appendix~\ref{App2} identifies an
induction-compatible setting for the path-line construction, without
assuming that roughness alone guarantees spontaneous stochasticity.

\appendix

\section{Topology Change and Spontaneous Stochasticity}\label{App0}

The phase-space interpretation proposed in Ref.~\cite{Jafari2025}
motivates a distinction between trajectory deformation, volume
compression, and magnetic topology change. Along a coarse-grained
magnetic path-line,
\[
\dot{\mathbf{X}}=\mathbf{B}_l(\mathbf{X},t),
\]
the chain rule gives
\[
\frac{d}{dt}\mathbf{B}_l(\mathbf{X}(t),t)
=\left[\partial_t\mathbf{B}_l+
(\mathbf{B}_l\cdot\nabla)\mathbf{B}_l\right]_{\mathbf{X}(t)}.
\]
The induction equation supplies the Eulerian time derivative, not this
total derivative. It does not by itself close a finite-dimensional
system in the values $(\mathbf{X},\mathbf{B}_l)$; spatial derivatives
and the velocity field also enter.

For a separately specified smooth closed model $\dot Z=F(Z,t)$,
$V(t)$ denotes an infinitesimal phase-space volume, and the divergence
below is evaluated along the corresponding trajectory $Z(t)$.
Liouville's identity gives
\[
\frac{d}{dt}\ln V=\nabla_Z\cdot F=\operatorname{tr}J_l,
\qquad J_l=\frac{\partial F}{\partial Z}.
\]
This measures signed phase-space volume change, not magnetic topology
change. For example, $\dot Z=-Z$ contracts volume but its finite-time
map $Z(0)\mapsto e^{-t}Z(0)$ is a diffeomorphism. A topology-change
criterion therefore requires a separate physical and mathematical
argument.

The inertial-range estimate
$|\mathbf{R}_l|\sim\delta u(l)\delta B(l)$ for the subscale
electromotive force motivates a characteristic rate
\[
\tau_T^{-1}(l)\sim\frac{\delta u(l)}{l}.
\]
For Kolmogorov scaling this is
$\tau_T^{-1}(l)\sim\varepsilon^{1/3}l^{-2/3}$. We retain it only as a
heuristic rate, not as the signed trace of a specified Jacobian. Its
growth does not by itself establish divergent Lyapunov exponents,
nonunique trajectories, or the pair-dispersion hypothesis \eqref{disp}.
The latter remains a separate assumption throughout the paper.

\section{Statistical Flux Freezing via Magnetic Path-Lines}\label{AppD}

In this appendix we derive a stochastic flux-freezing relation in
terms of magnetic path-lines. To avoid notational confusion, we denote
by $\mathbf{V}(\mathbf{x},t)$ the divergence-free transport field
generating the path-lines, and by $\mathbf{H}(\mathbf{x},t)$ the
divergence-free vector field transported by that flow. Here
$\mathbf{V}$ and $\mathbf{H}$ are conceptually distinct:
$\mathbf{V}$ generates the stochastic flow of trajectories, while
$\mathbf{H}$ denotes a vector field transported by that flow.

The derivation below establishes a general kinematic result: whenever a
divergence-free vector field $\mathbf{H}$ satisfies the transport
equation
\[
\partial_t \mathbf{H}
=
\nabla\times(\mathbf{V}\times\mathbf{H})
+
\kappa \Delta \mathbf{H},
\]
the flux of $\mathbf{H}$ through surfaces evolves according to a
stochastic flux-freezing relation associated with the stochastic flow
generated by $\mathbf{V}$. This statement does not require the
transport field $\mathbf{V}$ to coincide with the magnetic field or
with the fluid velocity.

For the physical MHD application, take $\mathbf{H}=\mathbf{B}$,
$\mathbf{V}=\mathbf{z}^{\pm}$, and $\kappa=\eta$, as specified in
Sec.~\ref{sec:flux}. Since
$\mathbf{z}^{\pm}\times\mathbf{B}=\mathbf{u}\times\mathbf{B}$,
Eq.~\eqref{AppD:induction} is exactly the resistive induction equation.
Taking $\mathbf{V}=\mathbf{u}$ recovers the representation of
Ref.~\cite{Eyink2011}. The fields are prescribed independently of the
auxiliary Brownian noise. A choice $\mathbf{V}=\mathbf{B}_l$ would
instead require the additional condition
$\nabla\times[(\mathbf{u}-\mathbf{B}_l)\times\mathbf{B}]=0$
for the same unfiltered source-free flux law.

Assume that $\mathbf{V}\in C_b^2(\mathbb{R}^3\times[t_0,t_f])$ with
$\nabla\cdot\mathbf{V}=0$, and let
$\mathbf{H}_0\in C_b^1(\mathbb{R}^3)$ satisfy
$\nabla\cdot\mathbf{H}_0=0$. 
Under these assumptions, for each fixed $\kappa>0$ the backward SDE
introduced below generates a stochastic flow of $C^1$ diffeomorphisms,
so that the Jacobian matrices and surface pullback used in the
derivation are well defined. For $\kappa>0$ consider the induction
equation driven by $\mathbf{V}$,
\begin{equation}\label{AppD:induction}
\partial_t \mathbf{H}
=
\nabla\times(\mathbf{V}\times\mathbf{H})
+
\kappa\Delta\mathbf{H},
\end{equation}
with $\nabla\cdot\mathbf{H}=0$ and
$\mathbf{H}(\mathbf{x},t_0)=\mathbf{H}_0(\mathbf{x})$. Using
$\nabla\cdot\mathbf{V}=\nabla\cdot\mathbf{H}=0$, this may be written
equivalently as
\begin{equation}
\partial_t \mathbf{H}
+
(\mathbf{V}\cdot\nabla)\mathbf{H}
=
(\mathbf{H}\cdot\nabla)\mathbf{V}
+
\kappa\Delta\mathbf{H}.
\label{AppD:induction2}
\end{equation}

Let $\widetilde{\mathbf{A}}_{s,t}(\mathbf{x})$ denote the backward
stochastic flow generated by $\mathbf{V}$, defined for $t_0\le s\le t$
by
\begin{equation}
\begin{aligned}
\hat d\widetilde{\mathbf{A}}_{s,t}(\mathbf{x})
&=\mathbf{V}(\widetilde{\mathbf{A}}_{s,t}(\mathbf{x}),s)\,ds
 +\sqrt{2\kappa}\,\hat d\mathbf{W}_s,\\
\widetilde{\mathbf{A}}_{t,t}(\mathbf{x})&=\mathbf{x}.
\end{aligned}
\label{AppD:backwardSDE}
\end{equation}
where $\hat d$ denotes the backward It\^o differential. Because the
drift field $\mathbf{V}$ is divergence-free and the noise is additive,
the stochastic flow generated by \eqref{AppD:backwardSDE} is
volume-preserving almost surely. For each fixed realization of the
Brownian motion, the map
$\mathbf{x}\mapsto\widetilde{\mathbf{A}}_{s,t}(\mathbf{x})$ is a
$C^1$ diffeomorphism, with inverse
$\widetilde{\mathbf{X}}_{s,t}=\widetilde{\mathbf{A}}_{s,t}^{-1}$.
The inverse map is the forward stochastic flow,
\[
\begin{aligned}
d\widetilde{\mathbf{X}}_{s,t}(\mathbf{a})
&=\mathbf{V}(\widetilde{\mathbf{X}}_{s,t}(\mathbf{a}),t)\,dt
 +\sqrt{2\kappa}\,d\mathbf{W}_t,\\
\widetilde{\mathbf{X}}_{s,s}(\mathbf{a})&=\mathbf{a}.
\end{aligned}
\]
Within a realization the same Brownian motion is used for every spatial
label. Independent Brownian realizations are used when sampling the
pair-dispersion law; this is distinct from independently perturbing
points of a single transported surface.

Write
\[
\mathbf{J}_{s,t}(\mathbf{a})
=
\nabla_{\mathbf{a}}\widetilde{\mathbf{X}}_{s,t}(\mathbf{a})
\]
for the Jacobian matrix of the forward flow, with
$(\nabla\mathbf{V})_{ij}=\partial_j V_i$. Differentiating the flow
equation with respect to $\mathbf{a}$ gives
\begin{equation}
\frac{d}{dt}\mathbf{J}_{s,t}(\mathbf{a})
=
\nabla\mathbf{V}\big(\widetilde{\mathbf{X}}_{s,t}(\mathbf{a}),t\big)
\,\mathbf{J}_{s,t}(\mathbf{a}),
\qquad
\mathbf{J}_{s,s}(\mathbf{a})=\mathbf{I}.
\label{AppD:J}
\end{equation}

Since $\nabla\cdot\mathbf{V}=0$, Liouville's formula implies
\begin{equation}
\det \mathbf{J}_{s,t}(\mathbf{a}) = 1.
\label{AppD:detJ}
\end{equation}

Define the random vector field
\begin{equation}
\widetilde{\mathbf{H}}(\mathbf{x},t)
=
\mathbf{J}_{t_0,t}(\mathbf{a})\,\mathbf{H}_0(\mathbf{a})
\Big|_{\mathbf{a}=\widetilde{\mathbf{A}}_{t_0,t}(\mathbf{x})}.
\label{AppD:Htilde}
\end{equation}

Equivalently,
\begin{equation}
\widetilde{\mathbf{H}}\big(\widetilde{\mathbf{X}}_{t_0,t}(\mathbf{a}),t\big)
=
\mathbf{J}_{t_0,t}(\mathbf{a})\,\mathbf{H}_0(\mathbf{a}).
\label{AppD:Htilde2}
\end{equation}

A forward It\^o--Wentzell calculation, interpreted weakly in space,
gives the stochastic transport equation
\begin{equation}
\begin{aligned}
d\widetilde{\mathbf{H}}
+\big[(\mathbf{V}\cdot\nabla)\widetilde{\mathbf{H}}
 -(\widetilde{\mathbf{H}}\cdot\nabla)\mathbf{V}
 -\kappa\Delta\widetilde{\mathbf{H}}\big]dt\\
+\sqrt{2\kappa}\,(d\mathbf{W}_t\cdot\nabla)\widetilde{\mathbf{H}}=0.
\end{aligned}
\label{AppD:SPDE}
\end{equation}
with initial condition
$\widetilde{\mathbf{H}}(\mathbf{x},t_0)=\mathbf{H}_0(\mathbf{x})$.

Taking expectations and using
$\mathbb{E}[(d\mathbf{W}_t\cdot\nabla)\widetilde{\mathbf{H}}]=0$
gives
\[
\partial_t \mathbb{E}\,\widetilde{\mathbf{H}}
+
(\mathbf{V}\cdot\nabla)\mathbb{E}\,\widetilde{\mathbf{H}}
=
(\mathbb{E}\,\widetilde{\mathbf{H}}\cdot\nabla)\mathbf{V}
+
\kappa\Delta \mathbb{E}\,\widetilde{\mathbf{H}},
\]
which coincides with \eqref{AppD:induction2}. By uniqueness of
solutions of \eqref{AppD:induction}, it follows that
\begin{equation}
\mathbf{H}(\mathbf{x},t)
=
\mathbb{E}\!\left[
\mathbf{J}_{t_0,t}(\mathbf{a})\,\mathbf{H}_0(\mathbf{a})
\Big|_{\mathbf{a}=\widetilde{\mathbf{A}}_{t_0,t}(\mathbf{x})}
\right].
\label{AppD:Lundquist}
\end{equation}

Equation \eqref{AppD:Lundquist} is the stochastic Lundquist formula in
path-line form.

We now derive the corresponding stochastic Alfv\'en theorem. Let
$S$ be a compact oriented smooth surface at time $t$, and let
$\widetilde{\mathbf{A}}_{t_0,t}(S)$ denote its random backward image
under the stochastic path-line flow. If
$d\mathbf{A}_{\mathbf{x}}$ and $d\mathbf{A}_{\mathbf{a}}$ are the
oriented area elements on $S$ and $\widetilde{\mathbf{A}}_{t_0,t}(S)$,
respectively, then the standard area transformation formula gives
\[
d\mathbf{A}_{\mathbf{x}}
=
\operatorname{cof}\mathbf{J}_{t_0,t}(\mathbf{a})\,d\mathbf{A}_{\mathbf{a}},
\]
where $\operatorname{cof}\mathbf{J}$ denotes the cofactor matrix.
Because of \eqref{AppD:detJ}, incompressibility implies
\[
\operatorname{cof}\mathbf{J}_{t_0,t}
=
\mathbf{J}_{t_0,t}^{-T}.
\]

Hence
\[
\big(\mathbf{J}_{t_0,t}\mathbf{H}_0\big)\cdot d\mathbf{A}_{\mathbf{x}}
=
\mathbf{H}_0\cdot d\mathbf{A}_{\mathbf{a}}.
\]

The bounded spatial gradient of $\mathbf{V}$ bounds the flow
Jacobian and inverse Jacobian on a finite time interval. Together with
bounded $\mathbf{H}_0$ and finite surface area, this justifies
interchanging expectation and surface integration. Using the area
identity in \eqref{AppD:Lundquist} yields
\begin{equation}
\int_S \mathbf{H}(\mathbf{x},t)\cdot d\mathbf{A}_{\mathbf{x}}
=
\mathbb{E}\!\left[
\int_{\widetilde{\mathbf{A}}_{t_0,t}(S)}
\mathbf{H}_0(\mathbf{a})\cdot d\mathbf{A}_{\mathbf{a}}
\right].
\label{AppD:Alfven}
\end{equation}

Equation \eqref{AppD:Alfven} is the stochastic Alfv\'en theorem in
path-line form: the flux through a surface at time $t$ equals the
ensemble average of the initial flux through the random surfaces
obtained by stochastic backward advection along path-lines.

The physical specialization and the additional assumptions for passing
\eqref{AppD:Alfven} to a singular ideal limit are those of
Sec.~\ref{sec:flux}. A fixed smooth drift has a deterministic zero-noise
limit; persistent noncollapse requires a rough drift or a simultaneous
rough-field limit. The present finite-diffusivity proof does not supply
uniform control of random surface fluxes in that singular limit.

\section{Stochastic path-lines}\label{App1}

In this appendix we summarize the mathematical assumptions and provide
the detailed arguments underlying the stochastic path-line formulation
used in the main text. We work in $\mathbb{R}^3$.

Throughout this appendix we consider $0<\kappa\le\kappa_0$ on a
finite time interval and assume:

(i) the drift $\mathbf{B}(\mathbf{x},t)$ is bounded and measurable;
if it depends on $\kappa$, its bound is uniform in $\kappa$ and that
dependence is suppressed in the regularized equations;

(ii) $\nabla\cdot\mathbf{B}=0$;

(iii) for each $\kappa>0$, Eq.~\eqref{SDE_app} with prescribed terminal
value has a weak Markov solution with backward transition density
$p^\kappa(\tau,\mathbf{y}|t,\mathbf{x})$;

(iv) the terminal-value path laws
\[
\mathrm{Law}(\mathbf{X}_\cdot^\kappa\mid
\mathbf{X}_t^\kappa=\mathbf{x})
\]
are tight on $C([\tau,t];\mathbb{R}^3)$ in the uniform topology.

For identifying limiting histories as solutions of
Eq.~\eqref{detpath_app}, additionally assume that the fixed drift is
continuous, or that the regularized drifts converge locally uniformly
in spacetime to a continuous limiting field $\mathbf{B}$. This last
condition is used in the path-space corollary, not in the endpoint
noncollapse proof.

Uniform boundedness of the drift and the Brownian increment estimates
provide tightness on finite time intervals in this setting. Weak
endpoint convergence is understood against bounded continuous test
functions. Conditioning notation denotes the prescribed terminal-value
law, not a Brownian bridge from an independently specified initial
particle distribution.

\subsection{Deterministic path-line dynamics}

We consider trajectories defined by
\begin{equation}
\dot{\mathbf{X}}(t) = \mathbf{B}(\mathbf{X}(t),t),
\label{detpath_app}
\end{equation}
where $\mathbf{B}(\mathbf{x},t)$ is a divergence--free vector field on
$\mathbb{R}^3$,
\[
\nabla\cdot \mathbf{B} = 0 .
\]

The endpoint noncollapse argument uses only the transition laws and the
specified pair-dispersion hypothesis. The additional assumptions above
justify the transport representation and the passage to path-space
limits; no dynamical closure for the drift is used.

If $\mathbf{B}$ is Lipschitz in space, the Picard--Lindel{\"o}f theorem
guarantees existence and uniqueness of solutions of \eqref{detpath_app}.
However, turbulent vector fields may violate this regularity
condition. For example, Kolmogorov scaling suggests spatial H{\"o}lder
continuity with exponent $1/3$, which lies below the Lipschitz
threshold. Such roughness is closely related to the Onsager conjecture
on anomalous dissipation and can lead to nonuniqueness of Lagrangian
trajectories.

The divergence--free condition implies that the advection operator
$\mathbf{B}\cdot\nabla$ is formally skew-adjoint in $L^2$ and that, for sufficiently regular fields, the associated deterministic flow is
volume preserving.

\subsection{Stochastic regularization}

To regularize the dynamics we introduce a small stochastic perturbation.
For $\kappa>0$ consider the backward stochastic differential equation
\begin{equation}
\begin{aligned}
\widehat d\mathbf{X}_s^\kappa
&=\mathbf{B}(\mathbf{X}_s^\kappa,s)\,ds
+\sqrt{2\kappa}\,\widehat d\mathbf{W}_s,\\
\mathbf{X}_t^\kappa&=\mathbf{x}.
\end{aligned}
\label{SDE_app}
\end{equation}
Here $\widehat d$ denotes the backward It\^o differential, with
physical time $s$ decreasing from $t$ and terminal value
$\mathbf{X}_t^\kappa=\mathbf{x}$. Equivalently, in increasing reversed
time $r=t-s$, $\mathbf{Y}_r=\mathbf{X}_{t-r}^\kappa$ satisfies
\[
d\mathbf{Y}_r=-\mathbf{B}(\mathbf{Y}_r,t-r)\,dr
+\sqrt{2\kappa}\,d\overline{\mathbf{W}}_r,
\qquad\mathbf{Y}_0=\mathbf{x},
\]
where $\overline{\mathbf{W}}_r$ is standard Brownian motion. This
fixes the sign convention without reversing physical time twice.
The parameter $\kappa$ is a regularizing diffusivity; it is identified
with magnetic diffusivity $\eta$ only in Sec.~\ref{sec:flux}.

The backward transition operators evolve in their terminal time with
the infinitesimal operator
\begin{equation}
L_s \phi(\mathbf{x})
=
- \mathbf{B}(\mathbf{x},s)\cdot\nabla \phi(\mathbf{x})
+
\kappa \Delta \phi(\mathbf{x}),
\label{generator_app}
\end{equation}
for smooth test functions $\phi$.

\subsection{Stochastic transport representation}

For $f\in C_b^2(\mathbb{R}^3)$ define
\begin{equation}
u^\kappa(\mathbf{x},t)
=
\int f(\mathbf{y})\,p^\kappa(\tau,\mathbf{y}\,|\,t,\mathbf{x})\,d\mathbf{y}.
\label{rep_app}
\end{equation}

\begin{proposition}[Backward Kolmogorov equation]
\label{prop:kol_app}
The function $u^\kappa$ defined in \eqref{rep_app} satisfies
\begin{equation}
\begin{aligned}
\partial_t u^\kappa+\mathbf{B}\cdot\nabla u^\kappa
&=\kappa\Delta u^\kappa,\\
u^\kappa(\mathbf{x},\tau)&=f(\mathbf{x}).
\end{aligned}
\label{kol_app}
\end{equation}
\end{proposition}

\begin{proof}
Write $P^\kappa_{t,\tau}f=u^\kappa(\cdot,t)$. For smooth
coefficients the short backward step gives
\[
P^\kappa_{t+h,t}\varphi=\varphi+hL_t\varphi+o(h).
\]
Together with the Chapman--Kolmogorov identity
$P^\kappa_{t+h,\tau}=P^\kappa_{t+h,t}P^\kappa_{t,\tau}$, this yields
\[
\partial_tP^\kappa_{t,\tau}f=L_tP^\kappa_{t,\tau}f.
\]
Substitution of \eqref{generator_app} gives \eqref{kol_app}, with
initial data from \eqref{rep_app}. For bounded measurable coefficients
the identity holds in the distributional sense, obtained by parabolic
approximation of the drift.
\end{proof}

Equation \eqref{kol_app} is an advection--diffusion equation with drift
$\mathbf{B}$ and diffusivity $\kappa$, posed forward in $t$ with
initial data at $t=\tau$. Backward trajectories therefore represent a
forward parabolic initial-value solution.

\subsection{Pair dispersion}

Consider two independent realizations of \eqref{SDE_app}, conditioned
to arrive at the same spacetime point $(\mathbf{x},t)$:
\[
\mathbf{X}_s^{\kappa,1}, \qquad \mathbf{X}_s^{\kappa,2}.
\]

Independence of the two Brownian noises is convenient but not
essential; the argument below requires only convergence of the
individual marginals.

\subsection{Noncollapse of the zero--noise limit}

\begin{theorem}[Non--Dirac zero--noise limit]
\label{thm:main_app}
Suppose there exist $\tau<t$, $r_*>0$, and $p_*>0$ such that
\begin{equation}
\liminf_{\kappa\to0}
\mathbb P
\big(
|\mathbf{X}_\tau^{\kappa,1}-\mathbf{X}_\tau^{\kappa,2}|
\ge r_*
\big)
\ge p_* .
\label{disp_app}
\end{equation}
Every weak subsequential endpoint limit is non-Dirac. In particular,
the transition kernel cannot converge weakly to a Dirac measure:
\[
p^\kappa(\tau,\cdot\,|\,t,\mathbf{x})
\nRightarrow
\delta_{\mathbf{a}_*}.
\]
Here weak convergence refers to convergence of probability measures on
$\mathbb{R}^3$ against bounded continuous test functions.
\end{theorem}

\begin{proof}
Suppose some sequence $\kappa_n\downarrow0$ satisfies
\[
p^{\kappa_n}(\tau,\cdot|t,\mathbf{x})
\Rightarrow\delta_{\mathbf{a}_*}.
\]
For $i=1,2$, concentration of the marginal gives
$\mathbb{P}(|\mathbf{X}^{\kappa_n,i}_\tau-\mathbf{a}_*|
\ge r_*/2)\to0$. The triangle inequality and union bound imply
\[
\begin{aligned}
&\mathbb{P}\big(|\mathbf{X}^{\kappa_n,1}_\tau-
\mathbf{X}^{\kappa_n,2}_\tau|\ge r_*\big)\\
&\quad\le\sum_{i=1}^2\mathbb{P}\big(
|\mathbf{X}^{\kappa_n,i}_\tau-\mathbf{a}_*|\ge r_*/2\big)
\longrightarrow0.
\end{aligned}
\]
This contradicts \eqref{disp_app} along that sequence. Hence no
subsequential limit is Dirac. Tightness supplies weak subsequential
limits, but their uniqueness is not asserted.
\end{proof}

\subsection{Stochastic path-line freezing}

By tightness, subsequences of the terminal-value path laws converge to
probability measures $\mu_{\mathbf{x},t}$ on
$C([\tau,t];\mathbb{R}^3)$. To identify their support, use the
integral form of \eqref{SDE_app}:
\[
\begin{aligned}
&\sup_{\tau\le s\le t}\left|\mathbf{X}^\kappa_s-\mathbf{x}
+\int_s^t\mathbf{B}(\mathbf{X}^\kappa_r,r)\,dr\right|\\
&\qquad\le\sqrt{2\kappa}\sup_{\tau\le s\le t}|\mathbf{W}_s-\mathbf{W}_t|.
\end{aligned}
\]
The right-hand side tends to zero in probability. Continuity of the
limiting drift, together with local uniform convergence when the drift
varies, permits passage to the limit in the path integral. Thus each
limit is supported on solutions of \eqref{detpath_app} with terminal
value $\mathbf{x}$.

\begin{corollary}[Stochastic path-line freezing]
Under the path-space tightness and drift-continuity assumptions, if
\eqref{disp_app} holds then every subsequential path-space limit
$\mu_{\mathbf{x},t}$ is non-Dirac and supported on admissible
zero-noise histories. Uniqueness of this measure is not asserted.
\end{corollary}

Indeed, evaluation at time $\tau$ is continuous in the uniform path
topology. The marginal of every path-space limit is therefore a
non-Dirac endpoint limit by Theorem~\ref{thm:main_app}. This establishes
trajectory stochasticity, not a magnetic-flux law without induction
compatibility.

\subsection{Remarks}

\textbf{Quantitative noncollapse.}
A lower bound
\[
\mathrm{Var}(\mathbf{X}_\tau^\kappa\mid
\mathbf{X}_t^\kappa=\mathbf{x})\ge c>0
\]
excludes Dirac subsequential limits provided squared endpoints are
uniformly integrable. Here variance means the trace of the covariance
matrix. Uniformly bounded drifts on a finite interval give uniformly
bounded fourth endpoint moments for $0<\kappa\le\kappa_0$, which
supply the required integrability. Thus the limiting backward endpoints
retain a finite spread whenever this variance bound holds.

\textbf{Two--particle formulation.}
For independent samples the two-particle kernel is the product
\[
\begin{aligned}
p_2^\kappa(\tau,\mathbf{y}_1,\mathbf{y}_2|t,\mathbf{x})
={}&p^\kappa(\tau,\mathbf{y}_1|t,\mathbf{x})\\
&\times p^\kappa(\tau,\mathbf{y}_2|t,\mathbf{x}).
\end{aligned}
\]
Condition~\eqref{disp_app} is the positive $\liminf$ bound on this
kernel integrated over $|\mathbf{y}_1-\mathbf{y}_2|\ge r_*$. In
particular, no sequence $\kappa_n\downarrow0$ can satisfy
\[
p_2^{\kappa_n}(\tau,\cdot,\cdot|t,\mathbf{x})
\Rightarrow\delta_{(\mathbf{a}_*,\mathbf{a}_*)}.
\]

\textbf{Path--space topology.}
Limits are taken in the uniform topology on
$C([\tau,t];\mathbb{R}^3)$. A one-point path law does not by itself
construct a limiting stochastic flow acting on entire surfaces; the
flux-limit assumptions in Sec.~\ref{sec:flux} are separate.

\section{Els{\"a}sser path-line formulation}\label{App2}

The stochastic path-line formalism developed in the main text and
Appendix~\ref{App1} can be extended to trajectory dynamics generated
by the Els{\"a}sser fields of incompressible magnetohydrodynamics.
Introducing the Els{\"a}sser variables
\[
\mathbf{z}^\pm = \mathbf{u} \pm \mathbf{B} ,
\]
for the following displayed evolution equation assume $\nu=\eta$.
The incompressible MHD equations may then be written in Els{\"a}sser form as
\[
\partial_t \mathbf{z}^\pm + (\mathbf{z}^\mp \cdot \nabla) \mathbf{z}^\pm
=
-\nabla P_* + \nu \Delta \mathbf{z}^\pm ,
\]
with $\nabla\cdot \mathbf{z}^\pm = 0$. Here $P_*=p/\rho+|\mathbf{B}|^2/2$ is the total
pressure per unit mass, $p$ is the fluid pressure, and $\nu$ is the
kinematic viscosity. Here $\mathbf{B}$ is in
Alfv\'en-speed units. The restriction $\nu=\eta$ is not needed for
the magnetic-flux identity in Sec.~\ref{sec:flux}.

In analogy with the magnetic path-line equation
$\dot{\mathbf{X}} = \mathbf{B}(\mathbf{X},t)$, one may consider the
Els{\"a}sser trajectory equations
\[
\dot{\mathbf{X}}^\pm(t) = \mathbf{z}^\pm(\mathbf{X}^\pm(t),t).
\]

These curves are generated by the Els\"asser vector fields. The
$\mathbf{z}^{\mp}$ field cross-advects $\mathbf{z}^{\pm}$ in the
Els\"asser equation, so the sign on $\mathbf{X}^{\pm}$ labels its
drift, not a fluctuation with the same sign that it carries. Pressure
and diffusion also remain in that equation. Roughness may permit
nonunique histories, but does not alone establish their occurrence.

A stochastic regularization can therefore be introduced:
\[
\begin{aligned}
\widehat d\mathbf{X}^{\pm,\kappa}_s
&=\mathbf{z}^\pm(\mathbf{X}^{\pm,\kappa}_s,s)\,ds
+\sqrt{2\kappa}\,\widehat d\mathbf{W}_s,\\
\mathbf{X}^{\pm,\kappa}_t&=\mathbf{x}.
\end{aligned}
\]

The same arguments used in Section~\ref{II} and Appendix~\ref{App1}
apply to this system. In particular, if backward pair dispersion
persists in the zero-noise limit, the transition kernels of the
stochastic Els{\"a}sser dynamics cannot collapse to Dirac measures.
Under the same assumptions, every subsequential endpoint limit is
non-Dirac. The dispersion hypothesis must be tested for the chosen
Els\"asser sign; it does not follow from dispersion of the magnetic
path-lines. For $\kappa=\eta$, these are also the trajectories in the
exact flux representation of Sec.~\ref{sec:flux}.

This extension indicates that spontaneous stochasticity is not tied
specifically to magnetic path-lines but can arise more generally in
trajectory dynamics generated by rough advecting fields appearing in
the MHD system. In particular, stochastic transport may arise in
multiple MHD propagation channels associated with the Els{\"a}sser fields.

\section{Feynman--Kac, stochastic trajectories, and magnetic-vector transport}
\label{sec:feynman-kac-review}

This appendix explains the scalar and vector stochastic representations
used above; the magnetic representation is that of
Ref.~\cite{Eyink2011}. We use bounded initial scalar data and smooth
bounded advection fields; the potential below is bounded and continuous.
The magnetic-vector construction uses the regularity assumptions of
Appendix~\ref{AppD}.

The purpose of this section is to clarify the relation between diffusion
terms in partial differential equations and stochastic trajectories.  The
main conceptual point is the following.  A diffusion equation can sometimes
be interpreted as a Fokker--Planck equation for a probability density, but
this is not the only, nor the most general, probabilistic interpretation.
Feynman--Kac is instead a representation theorem: it expresses the solution
of certain linear parabolic partial differential equations as an expectation
over stochastic paths.  The field being represented need not be a probability
density.  It may be a temperature, a concentration, a passive scalar, a vector
field, or a magnetic field.  The probability measure belongs to the stochastic
paths, not necessarily to the field itself.

This distinction is especially important for magnetic induction.  The magnetic
field is not a probability density.  Nevertheless, the resistive term
\(\lambda \Delta \boldsymbol{B}\) can be represented by Brownian wandering of
stochastic trajectories, while the stretching term
\((\boldsymbol{B}\cdot\nabla)\boldsymbol{u}\) is represented by a linear
transport equation for a vector carried along those trajectories.  Thus the
magnetic field at a spacetime point is represented not as the probability
density of a random object, but as an average of transported magnetic vectors
over an ensemble of stochastic histories.

\subsection{From Feynman path integrals to Feynman--Kac}

The analogy with the ordinary Feynman path integral is useful.  In quantum
mechanics, if the wavefunction at time \(0\) is \(\psi_0(x_i)\), then the
wavefunction at time \(t\) and final position \(x_f\) may be written as
\begin{equation}
    \psi(x_f,t)
    =
    \int K(x_f,t;x_i,0)\,\psi_0(x_i)\,dx_i ,
    \label{eq:qm-propagator}
\end{equation}
where \(K(x_f,t;x_i,0)\) is the quantum propagator.  Formally,
\begin{equation}
    K(x_f,t;x_i,0)
    =
    \int_{x(0)=x_i}^{x(t)=x_f}
    \mathcal{D}x\,
    \exp\!\left(\frac{i}{\hbar}S[x]\right).
    \label{eq:qm-path-integral}
\end{equation}
Thus the final wavefunction is obtained by summing over all initial points
\(x_i\), and over all paths connecting \(x_i\) to \(x_f\), with the oscillatory
weight \(\exp(iS/\hbar)\).

For the heat equation, the corresponding object is real and probabilistic
rather than complex and oscillatory.  Consider
\begin{equation}
    \partial_t c = \kappa \Delta c,
    \qquad
    c(x,0)=c_0(x).
    \label{eq:heat-equation}
\end{equation}
Here \(c_0(x)\) means the initial field at time \(0\).  It may be an initial
temperature profile, an initial dye concentration, or any other scalar
quantity.  It is not necessarily a probability density.

The solution is
\begin{equation}
    c(x_f,t)
    =
    \int_{\mathbb{R}^d}
    G_\kappa(x_f,t;x_i,0)\,c_0(x_i)\,dx_i,
    \label{eq:heat-kernel-solution}
\end{equation}
where
\begin{equation}
    G_\kappa(x_f,t;x_i,0)
    =
    \frac{1}{(4\pi\kappa t)^{d/2}}
    \exp\!\left[
        -\frac{|x_f-x_i|^2}{4\kappa t}
    \right]
    \label{eq:heat-kernel}
\end{equation}
is the heat kernel. As a function of $x_i$ it is the normalized density
of a backward Brownian ancestor with prescribed final point $x_f$.
This is not a Bayes-conditioned forward law with an arbitrary initial
particle distribution. The function $c_0(x_i)$ is the initial field
being averaged, not necessarily a probability density.

Equivalently, let \(W_t\) be a standard \(d\)-dimensional Brownian motion
defined on a probability space \((\Omega,\mathcal{F},\mathbb{P})\).  Brownian
motion has mean zero and covariance
\begin{equation}
    \mathbb{E}_{\mathbb{P}}\!\left[W_t^a W_t^b\right]
    =
    t\,\delta^{ab},
\end{equation}
where \(\mathbb{E}_{\mathbb{P}}\) denotes expectation with respect to Wiener
measure, i.e. the probability law of Brownian paths.  Then the random variable
\begin{equation}
    X_i = x_f + \sqrt{2\kappa}\,W_t
    \label{eq:random-ancestor-pure-diffusion}
\end{equation}
has density \(G_\kappa(x_f,t;x_i,0)\) as a function of \(x_i\).  Therefore
\begin{equation}
    c(x_f,t)
    =
    \mathbb{E}_{\mathbb{P}}
    \left[
        c_0\!\left(x_f+\sqrt{2\kappa}\,W_t\right)
    \right].
    \label{eq:heat-fk-simple}
\end{equation}
This formula is often the source of confusion because it appears to say:
``start at \(x_f\), walk randomly, and evaluate \(c_0\) where the walker
lands.''  That is mathematically correct, but the physical interpretation
should be stated carefully.  The point \(x_f\) is the fixed final point where
we want the solution.  The Brownian displacement samples possible ancestor
points \(x_i\) at time \(0\).  The value \(c(x_f,t)\) is the average of the
initial field \(c_0(x_i)\) over this Gaussian cloud of possible ancestors.

In other words, one should not imagine that the heat ``starts at \(x_f\)''
and then walks away.  Rather, to compute the value at \((x_f,t)\), diffusion
tells us to average over the values of the initial condition at all possible
earlier points that could have contributed to \(x_f\).  Brownian motion is
only a convenient way to sample those possible ancestors.

Thus the heat equation has two closely related, but conceptually distinct,
probabilistic interpretations.  First, if \(p(x,t)\) is the probability
density of Brownian particles satisfying
\begin{equation}
    dX_t = \sqrt{2\kappa}\,dW_t,
    \label{eq:brownian-sde}
\end{equation}
then \(p\) satisfies the Fokker--Planck equation
\begin{equation}
    \partial_t p = \kappa \Delta p.
    \label{eq:fp-pure-diffusion}
\end{equation}
In that special interpretation, the field \(p\) itself is a probability
density.

Second, if \(c\) is a temperature or scalar field satisfying the same heat
equation, then \(c\) need not be a probability density.  The same Brownian
paths can nevertheless be used to represent \(c\) by the expectation
\eqref{eq:heat-fk-simple}.  This is the basic Feynman--Kac idea: the solution
of a linear diffusion-type PDE can be written as an expectation over stochastic
paths.  The probability measure is the law of the paths; the field being
computed is the quantity averaged along, or over, those paths.

The analogy with the quantum path integral is then
\begin{equation}
    \text{quantum mechanics:}
    \qquad
    \int \mathcal{D}x\,
    \exp\!\left(\frac{i}{\hbar}S[x]\right),
\end{equation}
where paths carry complex amplitudes, whereas
\begin{equation}
    \text{Feynman--Kac:}
    \qquad
    \mathbb{E}_{\mathbb{P}}\!\left[\cdots\right],
\end{equation}
where paths are distributed according to a genuine probability measure,
typically Wiener measure or the law of an It\^o diffusion.  In imaginary
time, the oscillatory quantum weight is replaced by a real exponential
weight, and the path integral becomes closely related to Brownian expectation.

A slightly more general scalar example is
\begin{equation}
    \partial_t c
    =
    \kappa \Delta c - V(x,t)c,
    \qquad
    c(x,0)=c_0(x).
    \label{eq:heat-potential}
\end{equation}
For these bounded data, Feynman--Kac gives
\begin{equation}
    c(x,t)
    =
    \mathbb{E}_{x}
    \left[
        \exp\!\left(
            -\int_0^t V(X_s,t-s)\,ds
        \right)
        c_0(X_t)
    \right],
    \label{eq:fk-potential-forward}
\end{equation}
where $X_s=x+\sqrt{2\kappa}\,W_s$, and $\mathbb{E}_x$ denotes
expectation with $X_0=x$. The physical coefficient time is $t-s$
because this formula reconstructs an initial-value solution from its
value at the final time $t$.  The exponential
factor is a path-dependent weight.  If \(V>0\), paths are damped; if \(V<0\),
paths are amplified.  This is the diffusion analogue of a potential term in
imaginary-time quantum mechanics.

With advection, one considers the stochastic differential equation
\begin{equation}
    dX_s
    =
    \boldsymbol{u}(X_s,s)\,ds
    +
    \sqrt{2\kappa}\,dW_s.
    \label{eq:advection-diffusion-sde}
\end{equation}
The probability density of \(X_s\) satisfies the forward Fokker--Planck
equation
\begin{equation}
    \partial_t p
    =
    -\nabla\cdot(\boldsymbol{u}p)
    +
    \kappa\Delta p.
    \label{eq:fokker-planck-advection}
\end{equation}
For a fixed future time $t$, the observable
\begin{equation}
    q(x,s)=\mathbb{E}[f(X_t)\mid X_s=x],\qquad s\le t,
\end{equation}
satisfies a backward Kolmogorov equation in the starting time $s$,
\begin{equation}
\begin{aligned}
-\partial_s q(x,s)
&=\boldsymbol{u}(x,s)\cdot\nabla q(x,s)+\kappa\Delta q(x,s),\\
q(x,t)&=f(x).
\end{aligned}
\label{eq:backward-kolmogorov}
\end{equation}
This is the essential difference:
\begin{equation}
    \boxed{
    \begin{gathered}
\text{Fokker--Planck evolves}\\
\text{one-time particle probability densities,}
\end{gathered}
    }
\end{equation}
whereas
\begin{equation}
    \boxed{
    \begin{gathered}
\text{Feynman--Kac represents fields or observables}\\
\text{as expectations over paths.}
\end{gathered}
    }
\end{equation}
The same stochastic process can appear in both statements, but the objects
being evolved are different.

\subsection{Scalar fields, ancestor distributions, and the meaning of \(x\)}

It is useful to rewrite the scalar formula in explicitly backward language.
Suppose we want \(c(x_f,t_f)\).  The final point \(x_f\) and final time
\(t_f\) are fixed.  We then define a random ancestor \(A_{0,t_f}(x_f)\) at
time \(0\).  For pure diffusion,
\begin{equation}
    A_{0,t_f}(x_f)
    \overset{\mathrm{law}}{=}
    x_f+\sqrt{2\kappa}\,W_{t_f}.
    \label{eq:ancestor-pure-diffusion}
\end{equation}
The solution is
\begin{equation}
    c(x_f,t_f)
    =
    \mathbb{E}_{\mathbb{P}}
    \left[
        c_0\!\left(A_{0,t_f}(x_f)\right)
    \right].
    \label{eq:ancestor-expectation}
\end{equation}
Here the expectation is taken over all Brownian realizations \(\omega\in\Omega\),
or equivalently over all possible ancestor points \(A_{0,t_f}(x_f;\omega)\),
with probability determined by the Brownian transition kernel.

This makes the role of \(x_f\) clearer.  The point \(x_f\) is not being
forgotten.  It is the point about which the ancestor distribution is centered.
For small time \(t_f\), the ancestor cloud is narrow and concentrated near
\(x_f\).  For larger time, the ancestor cloud is broader.  The field value
\(c(x_f,t_f)\) is the average of \(c_0\) over that cloud.

For example, suppose
\begin{equation}
    c_0(x)
    =
    \begin{cases}
        100, & x<0,\\
        0, & x>0,
    \end{cases}
\end{equation}
in one spatial dimension.  Then
\begin{equation}
    c(x_f,t)
    =
    100\,
    \mathbb{P}\!\left(
        A_{0,t}(x_f)<0
    \right).
\end{equation}
The temperature at \(x_f\) and time \(t\) is \(100\) times the probability
that a backward Brownian ancestor of \((x_f,t)\) came from the initially hot
side.  The scalar \(c\) itself is not the probability density.  Rather, the
Brownian paths provide the probability distribution used to average the
initial scalar field.

In an advecting flow, the same idea holds, but the ancestor distribution is
no longer a simple Gaussian centered at \(x_f\).  If the scalar satisfies
\begin{equation}
    \partial_t c
    +
    \boldsymbol{u}\cdot\nabla c
    =
    \kappa \Delta c,
    \label{eq:scalar-advection-diffusion}
\end{equation}
then the backward trajectories have drift $\boldsymbol{u}$ and
diffusivity $\kappa$. With the same convention as
Eq.~\eqref{AppD:backwardSDE},
\begin{equation}
\begin{aligned}
\widehat d A_s
&=\boldsymbol{u}(A_s,s)\,ds+\sqrt{2\kappa}\,\widehat dW_s,\\
A_{t_f}&=x_f.
\end{aligned}
\label{eq:backward-ancestor-sde}
\end{equation}
where physical time $s$ decreases from $t_f$ to $0$ and the terminal
value is prescribed. Write $A_s=A_{s,t_f}(x_f)$ for this path. Then
\begin{equation}
    c(x_f,t_f)
    =
    \mathbb{E}
    \left[
        c_0(A_0)
        \,\middle|\,
        A_{t_f}=x_f
    \right],
    \label{eq:adv-diff-ancestor}
\end{equation}
where the expectation uses this prescribed terminal-value backward law,
not a forward diffusion conditioned from a separately specified prior.

In the limit \(\kappa=0\), there is no diffusion.  If the velocity field is
smooth enough to define unique trajectories, the ancestor distribution
collapses to a single deterministic point \(a\).  Then
\begin{equation}
    c(x_f,t_f)=c_0(a),
\end{equation}
where \(a\) is the unique initial point of the fluid trajectory that reaches
\(x_f\) at \(t_f\).  With diffusion, the single ancestor becomes a probability
distribution of ancestors, and the deterministic pullback of the initial
condition becomes an average over stochastic pullbacks.

Thus the difference between deterministic advection and advection--diffusion
is:
\begin{equation}
    \begin{gathered}
    \text{deterministic advection:}\\
    c(x_f,t_f)=c_0(\text{single ancestor}),
    \end{gathered}
\end{equation}
whereas
\begin{equation}
    \begin{gathered}
    \text{advection--diffusion:}\\
    c(x_f,t_f)=\mathbb{E}[c_0(\text{random ancestor})].
    \end{gathered}
\end{equation}
This is the scalar prototype for the magnetic-field representation.

\subsection{Magnetic induction as a vector Feynman--Kac representation}

We now turn to the magnetic case. Here $\lambda=\eta$ denotes magnetic
diffusivity, and $\kappa=\lambda$ for this magnetic representation.
For incompressible resistive MHD, the magnetic field satisfies
\begin{equation}
    \partial_t\boldsymbol{B}
    =
    \nabla\times(\boldsymbol{u}\times\boldsymbol{B})
    +
    \lambda \Delta \boldsymbol{B},
    \qquad
    \nabla\cdot\boldsymbol{u}=0,
    \qquad
    \nabla\cdot\boldsymbol{B}=0.
    \label{eq:resistive-induction}
\end{equation}
Using the vector identity
\begin{equation}
    \nabla\times(\boldsymbol{u}\times\boldsymbol{B})
    =
    (\boldsymbol{B}\cdot\nabla)\boldsymbol{u}
    -
    (\boldsymbol{u}\cdot\nabla)\boldsymbol{B}
    +
    \boldsymbol{u}(\nabla\cdot\boldsymbol{B})
    -
    \boldsymbol{B}(\nabla\cdot\boldsymbol{u}),
\end{equation}
and using \(\nabla\cdot\boldsymbol{u}=\nabla\cdot\boldsymbol{B}=0\), this becomes
\begin{equation}
    \partial_t\boldsymbol{B}
    +
    (\boldsymbol{u}\cdot\nabla)\boldsymbol{B}
    =
    (\boldsymbol{B}\cdot\nabla)\boldsymbol{u}
    +
    \lambda \Delta\boldsymbol{B}.
    \label{eq:induction-adv-stretch-diff}
\end{equation}
In components,
\begin{equation}
    \partial_t B_i
    +
    u_j\partial_j B_i
    =
    B_j\partial_j u_i
    +
    \lambda \Delta B_i.
    \label{eq:induction-components}
\end{equation}

This equation should be compared with scalar advection--diffusion,
\begin{equation}
    \partial_t c
    +
    \boldsymbol{u}\cdot\nabla c
    =
    \kappa\Delta c.
\end{equation}
The terms
\begin{equation}
    \boldsymbol{u}\cdot\nabla \boldsymbol{B}
    \quad\text{and}\quad
    \lambda\Delta\boldsymbol{B}
\end{equation}
have the same path interpretation as in the scalar case: the base point is
advected by \(\boldsymbol{u}\), and the diffusive term is represented by
Brownian wandering with amplitude \(\sqrt{2\lambda}\).  The new term is
\begin{equation}
    (\boldsymbol{B}\cdot\nabla)\boldsymbol{u}.
\end{equation}
This is the magnetic stretching and rotation term.  It says that the magnetic
vector is not merely carried as a passive scalar label.  It is deformed by the
velocity gradient along the trajectory.

In the ideal case \(\lambda=0\), along a deterministic plasma trajectory
\begin{equation}
    \frac{dX}{dt}
    =
    \boldsymbol{u}(X(t),t),
\end{equation}
the magnetic field satisfies
\begin{equation}
    \frac{d}{dt}\boldsymbol{B}(X(t),t)
    =
    (\boldsymbol{B}\cdot\nabla)\boldsymbol{u}(X(t),t).
\end{equation}
In components,
\begin{equation}
    \frac{dB_i}{dt}
    =
    B_j\,\partial_j u_i(X(t),t).
\end{equation}
Equivalently, if \(\boldsymbol{B}\) is treated as a column vector, this can be
written schematically as
\begin{equation}
    \frac{d\boldsymbol{B}}{dt}
    =
    \nabla\boldsymbol{u}(X(t),t)\,\boldsymbol{B},
    \label{eq:ideal-vector-transport}
\end{equation}
with $(\nabla\boldsymbol{u})_{ij}=\partial_j u_i$, as in
Appendix~\ref{AppD}.

For the resistive equation, define the forward stochastic flow
$X_r=\widetilde{\boldsymbol{X}}_{s,r}(a)$ with $X_s=a$ by
\begin{equation}
dX_r=\boldsymbol{u}(X_r,r)\,dr+\sqrt{2\lambda}\,dW_r,
\qquad X_s=a.
\label{eq:eyink-stochastic-path}
\end{equation}
The Brownian term represents the Ohmic diffusion term
$\lambda\Delta\boldsymbol{B}$. For a fixed realization, all spatial
labels use the same Brownian motion. The inverse map is
$A_{s,t}=\widetilde{\boldsymbol{X}}_{s,t}^{-1}$ and obeys the backward
SDE \eqref{eq:backward-ancestor-sde} with $\kappa=\lambda$.
Let $M_{s,t}(a)$ be the deformation matrix along the forward flow:
\begin{equation}
\begin{aligned}
\frac{d}{dt}M_{s,t}(a)
&=\nabla\boldsymbol{u}(\widetilde{\boldsymbol{X}}_{s,t}(a),t)
 M_{s,t}(a),\\
M_{s,s}(a)&=I.
\end{aligned}
\label{eq:deformation-matrix}
\end{equation}
Then Appendix~\ref{AppD}, with $\mathbf{V}=\mathbf{u}$, gives the
exact vector representation
\begin{equation}
\boldsymbol{B}(x_f,t_f)=\mathbb{E}\!\left[
M_{0,t_f}(a)\boldsymbol{B}_0(a)
\Big|_{a=A_{0,t_f}(x_f)}\right].
\label{eq:magnetic-vector-fk}
\end{equation}
The vertical bar means evaluation at the random inverse-flow label,
not conditioning a forward path with an unspecified initial law.
The matrix stretches and rotates the initial vector along the same
flow realization. The expectation is over the auxiliary Brownian noise
with the Eulerian fields held fixed.

This is the vector analogue of the scalar Feynman--Kac formula.  For a scalar
field, one averages initial scalar values:
\begin{equation}
c(x_f,t_f)=\mathbb{E}[c_0(A_{0,t_f}(x_f))].
\end{equation}
For the magnetic field, one averages transported initial vectors:
\begin{equation}
    \boldsymbol{B}(x_f,t_f)
    =
    \mathbb{E}
    \left[
        \text{stretched/rotated initial vector}
    \right].
\end{equation}
The magnetic field is not a probability density.  The probability measure is
the law of the stochastic trajectories.  The magnetic field is the vector
quantity being reconstructed by averaging over those trajectories.

A useful mnemonic is:
\begin{equation}
    \lambda\Delta\boldsymbol{B}
    \quad\longleftrightarrow\quad
    \begin{gathered}\text{Brownian wandering}\\\text{of the path position},\end{gathered}
\end{equation}
\begin{equation}
    \boldsymbol{u}\cdot\nabla\boldsymbol{B}
    \quad\longleftrightarrow\quad
    \text{advection of the base point},
\end{equation}
and
\begin{equation}
    (\boldsymbol{B}\cdot\nabla)\boldsymbol{u}
    \quad\longleftrightarrow\quad
    \begin{gathered}\text{stretching and rotation}\\\text{of the carried vector}.\end{gathered}
\end{equation}

This also clarifies the relation to stochastic flux-freezing.  In ordinary
ideal flux-freezing, a unique deterministic plasma trajectory is used to pull
back the magnetic field.  In resistive MHD, the Ohmic term produces a cloud of
stochastic ancestors.  The magnetic field at the final point is obtained by
transporting magnetic vectors along all those stochastic histories and then
averaging.  Thus the stochastic paths are not literal magnetic field lines,
and the magnetic field is not a probability density.  The paths provide the
probability measure over possible histories; the magnetic vectors are the
payload transported along those histories.

The same conceptual structure can be summarized as follows:
\begin{equation}
\begin{gathered}
\text{heat equation:}\\
\text{random ancestors + average initial numbers},
\end{gathered}
\end{equation}
\begin{equation}
\begin{gathered}
\text{scalar advection--diffusion:}\\
\text{random advective--diffusive ancestors}\\
+\;\text{average initial scalar labels},
\end{gathered}
\end{equation}
\begin{equation}
\begin{gathered}
\text{magnetic induction:}\\
\text{random advective--diffusive ancestors}\\
+\;\text{stretch/rotate initial magnetic arrows}\\
+\;\text{average arrows}.
\end{gathered}
\end{equation}
This is the meaning of the phrase ``vector Feynman--Kac representation.''  It
is not a Fokker--Planck interpretation of the magnetic field itself.  Rather,
it is a stochastic-path representation of the solution of a linear
vector-valued advection--diffusion--stretching equation.

\bibliography{main}
\end{document}